\documentclass[12pt,twoside,letterpaper]{llncs}
\usepackage{times}
\usepackage[T1]{fontenc}
\usepackage{amsmath}
\usepackage{amssymb}
\usepackage{fancybox}
\usepackage{graphics,graphicx}
\newtheorem{fact}{Fact}
\newcommand{\junk}[1]{}

\begin{document}
\title{Breaking the hegemony of the triangle method in clique detection}

\author{
Miros{\l}aw Kowaluk
\inst{1}
\and
Andrzej Lingas
\inst{2}
\institute{
  Institute of Informatics, University of Warsaw, Warsaw, Poland.
  \texttt{kowaluk@mimuw.edu.pl}
\and
  Department of Computer Science, Lund University,
22100 Lund, Sweden.
\texttt{Andrzej.Lingas@cs.lth.se}}
}
\maketitle

\begin{abstract} 
  We consider the fundamental problem of detecting/counting copies of
  a fixed pattern graph in a host graph. The recent progress on this
  problem has not included complete pattern graphs, i.e., cliques (and
  their complements, i.e., edge-free pattern graphs, in the induced
  setting). The fastest algorithms for the aforementioned patterns are
  based on a straightforward reduction to triangle
  detection/counting. We provide an alternative method of
  detection/counting copies of fixed size cliques based on a
  multi-dimensional matrix product. It is at least as time efficient
  as the triangle method in cases of $K_4$ and $K_5.$ The complexity
  of the multi-dimensional matrix product is of interest in its own
  rights. We provide also another alternative method for
  detection/counting $K_r$ copies, again time efficient for $r\in \{
  4,\ 5 \}$.
 \end{abstract}

\section{Introduction}
The problems of
detecting, finding, counting or listing
subgraphs or induced subgraphs 
of a host graph that are
isomorphic to a pattern graph 
are basic in graph algorithms. 
They are generally termed as {\em subgraph isomorphism}
and {\em induced subgraph isomorphism} problems, respectively.
Such well-known NP-hard problems
as the independent set, clique, Hamiltonian cycle or 
Hamiltonian path can be regarded as their special cases.
 
Recent examples of applications of different variants of subgraph
isomorphism include among other things \cite{FKL15,VW15}:
a comparison of bio-molecular networks by their so-called motifs
\cite{ADH08},
an analysis of social networks by counting the number of copies
of a small pattern graph
\cite{SW05},
graph matching constraints in
automatic design of processor systems
\cite{WK09},
and the detection of communication patterns
between intruders in network security \cite{SS04}.
In the
aforementioned applications, the pattern graphs are typically of fixed
size which allows for polynomial-time solutions.

At the beginning of 80s, Itai and Rodeh \cite{itai} presented
the following straightforward reduction of not only
triangle detection but also triangle counting to fast
matrix multiplication.
Let $A$ be the $0-1$ adjacency matrix of the host graph $G$ on $n$ vertices
(see Preliminaries).
Consider the matrix product $C=A\times A.$ Note that
$C[i,j]=\sum_{l=1}^n A[i,l]A[l,j]$ is the number of two-edge paths
connecting the vertices $i$ and $j.$ Hence, if $\{i,j\}$ is an edge
of $G$ then $C[i,j]$ is the number of triangles in $G$ including
the edge $\{i,j\}$. Consequently, the number of triangles
in an $n$-vertex graph can be reported in $O(n^{\omega})$
time, where $\omega$ is the exponent of fast matrix multiplication for 
matrices of size $n\times n$. (Recently, Alman and Vassilevska Williams
have shown that $\omega \le 2.3729$ \cite{AV21}.)

A few years later, Necetril and Poljak \cite{NP85}
showed an efficient reduction of detection and counting
copies of any pattern graph both in the standard
and induced case to the aforementioned method
for triangle detection and counting. The idea
is to divide the pattern graph into three almost
equal parts and to build an auxiliary graph on
copies of subgraphs isomorphic to one of three
parts. Then, the triangle detection/counting method
is run on the auxiliary graph. Two decades later,
Eisenbrand and Grandoni \cite{EG04} (cf. \cite{KKM00}) refined
this general triangle method by using fast
algorithms for rectangular matrix multiplication
instead of those for square matrix multiplication.
For a pattern graph on $r\ge 3$ vertices and a host graph on $n$ vertices,
the (refined) general triangle method runs
in time 
$O(n^{\omega(\lfloor r/3 \rfloor, \lceil (r-1)/3 \rceil,
  \lceil r/3 \rceil )})$ \cite{EG04,KKM00,NP85}, where $\omega (p,q,s)$
denotes the exponent of fast matrix multiplication for rectangular
matrices of size $n^p \times n^q$ and $n^q\times n^s$, respectively
\cite{LeGall12}. For example, it is known that  $\omega(1,2,1)\le 3.257 $ \cite{LeGall12}.

Up to now, the general triangle method remains the fastest known
universal method for the detection and counting standard and induced
copies of fixed pattern graphs.  In the recent two decades, there has
been a real progress in the design of efficient algorithms for
detection and even counting of fixed pattern graphs both in the
standard \cite{FKL15,KLL13} and induced case
\cite{BKS18,DDV19,FKL15,FKLL12,VW15}. Among other things, the progress has
been based on the use of equations between the numbers of copies of
different fixed patterns in the host graph
\cite{DDV19,KKM00,KLL13,VW15} and randomization
\cite{DDV19,FKL15,VW15}.  Unfortunately, this progress has not
included complete pattern graphs, i.e., $K_r$ graphs (and their complements, i.e., edge-free pattern
graphs in the induced setting). For the aforementioned
pattern graphs, the generalized triangle method remains the fastest
known one.

In this paper, we consider another universal method that
in fact can be viewed as another type of generalization
of the classic algorithm for triangle detection
and counting due to Itai and Rodeh. We can rephrase
the description of their algorithm as follows.
At the beginning, we form a list of subgraphs
isomorphic to $K_2$ (i.e., edges) and then
for each subgraph on the list we count the number of vertices
outside it that are adjacent to both vertices
of the subgraph, in other words, we count the number of extensions of
the subgraph to a clique on three vertices. The latter
task can be done efficiently by squaring the adjacency matrix
of the host graph. We can generalize the algorithm
to include detection/counting $K_r$ copies, $r\ge 3,$ by replacing
$K_2$ with $K_{r-1}$ and using a $(r-1)$-dimensional product
of $r-1$ copies of the adjacency matrix (see Section 3 for
the definition) instead of squaring
the matrix. Listing the subgraphs of the host graph
takes $O(n^{r-1})$ time so the overall time required
by the alternative method is $O(n^{r-1}+ n^{\omega_{r-1}}),$
where $\omega_{k}$ is the exponent of fast $k$-dimensional
product of $k$ $n\times n$ matrices.
On the other hand, we show in particular that $\omega_k\le \omega (\lceil
k/2 \rceil ,1,\lfloor k/2 \rfloor )$.
Hence, our alternative method in particular
computes the number of $K_4$ copies in an $n$-vertex
  graph in $O(n^{\omega(2,1,1)})$ time
  and the number of $K_5$ copies in $O(n^{\omega(2,1,2)})$ time.
  Also, if the input graph contains a copy
  of $K_4$ or $K_5$ respectively then
  a copy of $K_4$ can be found
  in the graph in $\tilde {O}(n^{\omega(2,1,1)})$ time
  while that of $K_5$
  in $\tilde{O}(n^{\omega(2,1,2)})$ time by a slightly
  modified alternative method.
  Thus, our upper time bounds for $K_4$ and $K_5$ at least match those
  for $K_4$ and $K_5$ yielded by the generalized triangle method
  \cite{EG04}. If $\omega_k<\omega (\lceil k/2
\rceil, 1, \lfloor k/2 \rfloor )$
for $k$ equal to $3$ or $4$ then we would get a breakthrough in detection/counting
of $K_4$ or $K_5,$ respectively.
For $K_r,$ where $r\ge 6,$ the generalized triangle method asymptotically
subsumes our alternative method and for $K_3$ the methods coincide.

We provide also another alternative
method for detection/counting $K_r$ copies, where $r\ge 3.$
It starts from listing all $K_{r-2}$ copies and then
it tries to extend them by two vertices to $K_r$ copies.
Again, the method is time efficient for $r\le 5$.
Finally, in order to obtain a method for detection/counting
$K_r$ copies that could compete with the generalized triangle method
for $r\ge 6$,
we consider a generalization of our alternative methods.
Similarly, it starts from listing all $K_q$ copies,
where $q<r-1,$ and then it tries to extend them by $r-q$
vertices to form $K_r$ copies. However, to perform
the extension step efficiently, we need to split
the extending vertex sets in two almost equal parts,
so the generalized method can be also regarded as a variant
of the generalized triangle one.

\subsection{Paper organization}

In the next section, the basic matrix and graph notation used in the
paper is presented.  Section 3 is devoted to the $k$-dimensional
matrix product of $k$ matrices, in particular its definition and upper
time bounds on the product in terms of those for fast rectangular
matrix multiplication. In Section 4, the alternative method for
detection/counting copies of fixed cliques in a host graph relying on
the multi-dimensional matrix product is presented and analyzed. Section 5
presents shortly another alternative method for detection/counting $K_r$
copies while Section 6 is devoted to a generalization
of the alternative methods.  We conclude with open problems.

\section{Preliminaries}

For a positive integer $r,$ we shall denote
the set of positive integers not greater than $r$
by $[r].$

For a matrix $D,$ $D^T$ denotes its transpose.  For positive real numbers
$p,\ q,\ s,$ $\omega(p,q,s)$ denotes the exponent of fast matrix
multiplication for rectangular matrices of size $n^p\times n^q$ and $n^q\times
n^s,$ respectively.
For convenience, $\omega =\omega(1,1,1).$

Let $\alpha$ stand
for  $sup \{0
\le q \le 1 : \omega(1,q,1) = 2 + o(1) \}.$
The following recent lower bound on
$\alpha$ is due to Le Gall and Urrutia \cite{LGU}.

\begin{fact}\label{fact: ur}
The inequality $\alpha >  0.31389$ holds \cite{LGU}.
\end{fact}

A \textit{witness} for a non-zero entry $C[i, j]$ of the Boolean matrix product
$C$ of a Boolean $p\times q$ matrix  $A$ and  a Boolean $q\times s$ matrix $B$ is any index $\ell\in [q]$
such that $A[i,\ell]$
and $B[\ell , j]$ are equal to 1.

The \textit{witness problem} is to report
a witness for each non-zero entry of the Boolean matrix product
of the two input matrices.

Alon and Naor provided a solution
to the witness problem for square Boolean matrices \cite{AN96}
which is almost equally fast as that for square matrix multiplication \cite{AV21}.
It can be easily generalized to include the Boolean product
of two rectangular Boolean matrices of sizes $n^p\times n^q$
and $n^q\times n^s,$ respectively.
The asymptotic matrix multiplication time $n^{\omega}$
is replaced by $n^{\omega(p,q,s)}$ in the generalization.

\begin{fact}
\label{fact: wit}
For positive $p,q,s,$
the witness problem for the Boolean matrix
product of an $n^p\times n^q$ Boolean matrix
with an $n^q\times n^s$ Boolean matrix
can be solved (deterministically) in $\tilde{O}(n^{\omega(p,q,s)})$
time.
\end{fact}

We shall consider only simple undirected graphs. 

A \textit{subgraph} of the graph $G = (V, E)$ is a graph $H = (V_H,
E_H)$ such that $V_H \subseteq V$ and $E_H \subseteq E$.

An \textit{induced subgraph} of the graph $G = (V, E)$ is a graph $H =
(V_H, E_H)$ such that $V_H \subseteq V$ and $E_H = E \cap (V_H \times
V_H)$. A subgraph of $G$ \textit{induced by} $S\subseteq V$ is a graph $F=(V_F,E_F)$
such that $V_F=S$ and $E_F = E \cap (S \times S)$. It is denoted by $G[S].$

For simplicity, we shall refer to a subgraph of a graph $G$ that is
isomorphic to $K_r$ as a \textit{copy} of $K_r$ in $G$ or just $K_r$ copy
in $G.$

The \textit{adjacency matrix} $A$ of a graph $G = (V, E)$ is the
$0-1$ $n\times n$ matrix such that $n=|V|$ and
for $1\le i,j\le n,$ $A[i, j]=1$ if and only if $\{ i, j\} \in E$.


\section{Multi-dimensional matrix product}

\begin{definition}
  For $k$ $n\times n$ 
  matrices $A_q$, $q=1,...,k,$ (arithmetic or Boolean, respectively) their
  $k$-dimensional (arithmetic or Boolean, respectively)
  matrix product $D$ is defined by
$$D[i_1,i_2,...,i_{k}]=
  \sum_{\ell=1}^n
  A_1[i_1, \ell]A_2[i_2, \ell]...A_{k}[i_{k},\ell],$$

  where $i_j\in [n]$ for $j=1,...,k.$
  The exponent of fast  $k$-dimensional (arithmetic) matrix product
  of $k$  $n\times n$ matrices is denoted by $\omega_k.$

  In the Boolean case,
  a witness for a non-zero entry $D[i_1,i_2,...,i_{k}]$ 
  of the $k$-dimensional Boolean matrix product
  is any index $\ell\in [n]$ such that 
  $A_1[i_1, \ell]A_2[i_2, \ell]...A_{k}[i_{k},\ell]$
  is equal to (Boolean) $1.$
  The witness problem for the
  $k$-dimensional Boolean matrix product is to report
  a witness for each non-zero entry of the product.
\end{definition}

Note that in particular the $2$-dimensional
matrix product of the matrices $A_1$ and $A_2$
coincides with the standard matrix product
of $A_1$ and $(A_2)^T$ which yields
$\omega_2=\omega .$

\begin{lemma} \label{lem: kdim}
  Let $k, k_1, k_2$ be three positive integers such that $k=k_1+k_2.$
  Both in the arithmetic
  and Boolean case, the $k$-dimensional matrix product of $k$  $n\times n$
  matrices can be computed  in $O(n^{\omega (k_1,1,k_2)})$ time, consequently
  $\omega_k\le \omega (k_1,1,k_2).$ Also, in the Boolean case,
  the witness problem for the $k$-dimensional matrix product
  can be solved in $\tilde{O}(n^{\omega (k_1,1,k_2)})$ time.
\end{lemma}
\begin{proof}
  To prove the first part, it is sufficient to consider the
  arithmetic case as the Boolean one trivially reduces to it.
  
  Let $A_1,....,A_{k}$ be the input matrices.
  Form an $n^{k_1}\times n$ matrix $A$ whose rows are
  indexed by $k_1$-tuples of indices in $[n]$
  and whose columns are indexed by indices in $[n]$
  such that $A[i_1...i_{k_1}, \ell ]=A_1[i_1,\ell]....A_{k_1}[i_{k_1}, \ell].$
Similarly, form an  $n^{k_2}\times n$ matrix $B$ whose rows are
  indexed by $k_2$-tuples of indices in $[n]$
  and whose columns are indexed by indices in $[n]$
  such that $B[j_1...j_{k_2}, \ell ]=A_{k_1+1}[j_1,\ell]....A_{k}[j_{k_2}, \ell].$
  Compute the rectangular matrix product $C$ of the matrix $A$ with
  the matrix $B^T$.
  By the definitions, the $D[i_1,...,i_{k_1},i_{k_1+1},...,i_k]$ entry
  of the product of the input matrices $A_1,....,A_{k}$ is equal
  to the entry $C[i_1...i_{k_1},i_{k_1+1}...i_k].$
    The matrices $A,\ B$ can be formed in $O(n^{k_1+1}+n^{k_2+1})$ time,
    i.e., $O(n^k)$ time, while the product $C$
    can be computed in $O(n^{\omega(k_1,1,k_2)})$ time.

    To prove the second part of the lemma it is sufficient to
    consider Boolean versions of the matrices $A,\ B,\ C$ and use
    Fact  \ref{fact: wit}.
      \qed
\end{proof}

  By combining Lemma \ref{lem: kdim} with Fact \ref{fact: ur}, we obtain the following corollary.
  
  \begin{corollary}
    For even $k\ge 8,$ $\omega_k=k+o(1).$
  \end{corollary}
  \begin{proof}
    We obtain the following chain of qualities
    on the asymptotic time required by the $k$-dimensional
    matrix product using Lemma \ref{lem: kdim} and Fact \ref{fact: ur}:
    $$n^{\omega(k/2,1,k/2)}=(n^{k/2})^{\omega(1,2/k,1)}=(n^{k/2})^{2+o(1)}=n^{k+o(1)}.$$
    \qed
  \end{proof}
  
\section{Clique detection}

The following algorithm is a straightforward
generalization of that due to Itai and Rodeh for
triangle counting \cite{itai} to include $K_r$ counting, for $r\ge 3$.
\newpage
\noindent
{\bf Algorithm 1}
\begin{enumerate}
\item form  a list $L$ of all $K_{r-1}$ copies in $G$
\item $t\leftarrow 0$
\item {\bf for} each  $C\in L$ {\bf do}
  \newline
  increase $t$ by the number of vertices in $G$ that are adjacent
  to all vertices of $C$
\item return $t/r$
\end{enumerate}

The correctness of Algorithm 1 follows from the fact that the number
of $K_r$ copies including a given copy $C$ of $K_{r-1}$ in the host
graph is equal to the number of vertices outside $C$ in the graph that
are adjacent to all vertices in $C$ and that a copy of $K_r$ includes
exactly $r$ distinct copies of $K_{r-1}$ in the graph.

The first step of Algorithm 1 can be implemented in $O(n^{r-1})$ time.
We can use the $(r-1)$-dimensional matrix product
to implement the third step
by using the next lemma immediately following from the definition
of the product.

\begin{lemma}\label{lem: kdap}
  Let $D$ be the $k$-dimensional matrix product of $k$ copies of
  the adjacency matrix of the input graph $G$ on $n$ vertices. Then,
  for any $k$ tuple $i_1,\ i_2,\ ...,\ i_{k}$ of vertices of $G$,
  the number of vertices in $G$ adjacent to each vertex in the $k$ tuple
  is equal to $D[i_1,i_2,...,i_{k}].$
  \end{lemma}

By the discussion and Lemma \ref{lem: kdap},
we obtain the following theorem.

\begin{theorem}\label{theo: first}
The number of $K_r$ copies in the input graph on $n$ vertices
can be computed (by Algorithm 1) in $O(n^{r-1}+n^{\omega_{r-1}})$ time.
\end{theorem}

By Lemma \ref{lem: kdim}, we obtain the following corollary
from Theorem \ref{theo: first},
matching the upper time bounds on the detection/counting
copies of $K_4$ and $K_5$
established in \cite{EG04}.

\begin{corollary}
  The number of $K_4$ copies in an $n$-vertex
  graph can be computed (by Algorithm 1) in $O(n^{\omega(2,1,1)})$ time
  while the number of $K_5$ copies in $O(n^{\omega(2,1,2)})$ time.
  Also, if the input graph contains a copy
  of $K_4$ or $K_5$ respectively then
  a copy of $K_4$ can be found
  in the graph in $\tilde {O}(n^{\omega(2,1,1)})$ time
  while that of $K_5$
  in $\tilde{O}(n^{\omega(2,1,2)})$ time (by a modification of Algorithm 1).
\end{corollary}

\section{Another alternative method for $K_r$ detection/counting}

The basic idea of our alternative method for
detection/counting copies of $K_r$
presented in the previous section
is to list copies of $K_{r-1}$ and then extend
them by single vertices to form copies of $K_r$
if possible. In this section, we present
a similar method based on the idea of extending
$K_{r-2}$ copies by pairs of vertices if possible.

This simple method for
detection/counting copies of $K_r$, where $r\ge 3,$ in a host graph $G=(V,E)$
on $n$ vertices
is as follows.  First,
we form a list $L$ of all $K_{r-2}$ copies in $O(n^{r-2})$ time.
Then, for each $H$ in $L$,
we compute the set $S(H)$ of vertices
which extend $H$ to a copy of $K_{r-1}$ in $G.$
It takes totally $O(n\times n^{r-2})$ time.
Next, we form a $0-1$ $n \times O(n^{r-2})$ matrix $B$
whose rows correspond to $v\in V$ and
whose columns correspond to $H\in L$ such
that $B[v,H]=1$ if and only if $v\in S(H).$
Then, we compute the matrix product $C$ of
$B$ with its transpose $B^T$ in $O(n^{\omega(1,r-2,1)})$ time.
Note that $C[v,u]$ is equal to the number of
copies of $K_{r-2}$ in $G$ that can be extended to a
copy of $K_{r-1}$ in $G$ both by $v$
and $u.$  Now, it is sufficient to check for each
non-zero entry $C[v,u]$ if in the adjacency matrix
$A$ of $G$, for the corresponding
entry $A[v,u]=1$ holds. Simply, then the pair of vertices
$v,\ u$ extending the same $C[v,u]$ copies of $K_{r-2}$
in $G$ to pairs of  $K_{r-1}$ copies in $G$ is adjacent
so $C[v,u]$ copies  of $K_r$ occur in $G.$
More concisely, we can
describe this method as follows
under the assumptions that $r\ge 3, $ $G=(V,E)$
is the input graph and $A$ is its adjacency matrix.
\par
\vskip 3pt
\noindent
{\bf Algorithm 2}
\begin{enumerate}
\item $L\leftarrow $ a list of all $K_{r-2}$ copies in $G$
\item {\bf for} $H\in L$ {\bf do}
  \newline
 $S(H)\leftarrow $ the set of vertices extending $H$ to a copy of $K_{r-1}$ in $G$
\item initialize a $0-1$ $|V| \times |L|$ matrix $B$
\item {\bf for} $v\in V \land H\in L$ {\bf do}
  \newline
      {\bf if} $v\in S(H)$ {\bf then} $B[v,H]\leftarrow 1$
        {\bf else} $B[v,H]\leftarrow 0$
\item $C\leftarrow B\times B^T$
\item $t\leftarrow 0$
\item {\bf for} $\{v,\ u\} \subset V$ {\bf do}
\newline
 {\bf if} $A[v,u]=1$  {\bf then}
$t\leftarrow t+C[v,u]$
 \item return $t/\binom r 2$
\end{enumerate}

As each edge $\{v,u\}\in E$ occurs in $C[v,u]$ copies
of $K_r$ in $G$ it contributes $C[v,u]$
to $t.$ On the other hand, $K_r$ has $\binom r 2$ edges.
Hence, the final value of $t$ divided by $\binom r 2$ yields
the number of $K_r$ copies in $G.$
By the discussion
and $\omega(1,r-2,1)\ge r-1$,
we obtain
the following theorem.

\begin{theorem}
  Algorithm 2 computes the number of $K_r$ copies  in an $n$-vertex
  graph in $O(n^{\omega(1,r-2,1)})$ time. 
\end{theorem}

\begin{corollary}
Algorithm 2 computes the number of $K_4$ copies in
 $O(n^{\omega(1,2,1)})$ time while the number of $K_5$ copies
  in $O(n^{\omega(1,3,1)})$ time.
\end{corollary}

Again, we can use Fact \ref{fact: wit} to modify Algorithm 2
to find a copy of $K_r$ in $\tilde{O}(n^{\omega(1,r-2,1)})$ time
in the graph in case it contains copies of $K_r.$

  \section{A generalization of the alternative methods}
  Our two alternative methods for detection/counting $K_r$ copies at
  least match the generalized triangle method for $r\le 5$ but they are
  asymptotically subsumed by the latter method for larger $r.$ In this
  section, we present a generalization of our two alternative methods
  that for appropriate parameters is competitive even for $r$ larger
  than $5.$ The basic idea of the generalization is to start from
  listing copies of $K_q$ in the host graph, where $r-q\ge 2$, and
  then to detect extensions of the $K_q$ copies by $r-q\ge 2$ vertices
  to $K_r$ copies in the graph. To perform the latter task efficiently,
  we split such an extension into two almost equal parts, so this generalized method
  can be also regarded as a variant of the triangle one.

The generalized method for
detecting  copies of $K_r$ in a host graph $G=(V,E)$
on $n$ vertices presented in this section
is as follows.  First,
we form a list $L$ of all $K_{q}$ copies in $O(n^{q})$ time.
Then, for each $H$ in $L$,
we compute the set $S(H)$ of vertices
which extend $H$ to a copy of $K_{q+1}$ in $G.$
It takes totally $O(n\times n^{q})$ time.
Now, to find extensions of the $K_q$ copies by $r-q$
vertices to form $K_r$ copies, we set $r_1$ to $\lceil \frac {r-q} 2 \rceil$
and $r_2$ to $\lfloor \frac {r-q} 2 \rfloor$.
Next, for each $H\in L$ and  $i\in [2],$  we form a list $L_i(H)$
of all $K_{r_i}$ copies in $G[S(H)]$.
It takes totally $O(n^{q+r_1})$ time.
Then, for $i\in [2],$ we create a $0-1$ matrix $B_i$
whose rows correspond to sets $s_i$ of $r_i$ vertices
and whose columns correspond to $H\in L$ such that
$B_i[s_i,H]=1$ if and only if there is a copy
of $K_{r_i}$, whose vertex set is $s_i$, in $L_i(H).$
Again, this  takes totally $O(n^{q+r_1})$ time.
Now, it is sufficient to compute the matrix product $C$ of
$B_1$ with $B_2^T$ and check
if there is a non-zero entry $C[s_1,s_2],$ where
$s_1\cup s_2$ induces a copy of $K_{r-q}$ in the
graph. Simply, then all vertices in the induced
$(r-q)$-clique have to be adjacent to the same
$H \in L,$ so they jointly with the vertices of $H$
induce a copy of $K_r$ in $G.$
The computation of the matrix product
$C$ takes $O(n^{\omega(r_1,q,r_2)})$ time,
and the checking of the matrix product $O(n^{r_1+r_2})$
time, i.e., $O(n^{r-q})$ time.

More concisely, we can
describe this method as follows
under the assumptions that $r\ge 3,$  $q\in [r-2],$ and $G=(V,E)$
is the input graph.
\par
\vskip 3pt
\noindent
{\bf Algorithm 3}
\begin{enumerate}
\item $L\leftarrow $ a list of all $K_{q}$ copies in $G$
\item {\bf for} $H\in L$ {\bf do}
  \newline
 $S(H) \leftarrow$ the set of vertices extending $H$ to a copy of $K_{q+1}$ in $G$
\item $r_1 \leftarrow \lceil \frac {r-q} 2 \rceil$
\item $r_2 \leftarrow \lfloor \frac {r-q} 2 \rfloor$
\item {\bf for} $H\in L \land i\in [2]$ {\bf do}
 \newline
$L_i(H)\leftarrow $ a list of all $K_{r_i}$ copies in $G[S(H)]$
\item {\bf for} $i\in [2]$ {\bf do}
\newline
form a $0-1$ matrix $B_i$ 
whose rows correspond to sets $s_i$ of $r_i$ vertices
and whose columns correspond to $H\in L$
such that $B_i[s_i,H]=1$ iff $G[s_i]$ is a copy of $K_{r_i}$ in $L_i(H)$
\item $C\leftarrow B_1\times B_2^T$
\item {\bf for} each $r_1$-vertex subset $s_1$ and each $r_2$-vertex subset $s_2$ {\bf do}
\newline
    {\bf if} $C[s_1,s_2]=1$ and $G[s_1\cup s_2]$ is a
    $(q-r)$-clique {\bf then} return YES and stop
 \item return NO
\end{enumerate}

By the discussion and $\omega(\lceil \frac {r-q} 2 \rceil,q,\lfloor
\frac {r-q} 2 \rfloor )\ge \lceil \frac {r-q} 2 \rceil +q $, we obtain
the following theorem.

\begin{theorem}
  Let $q\in [r-2].$
  Algorithm 3 detects a copy of $K_r$ in an $n$-vertex
  graph in $O(n^{\omega(\lceil \frac {r-q} 2 \rceil,q,\lfloor \frac {r-q} 2 \rfloor )})$ time. 
\end{theorem}
\junk{
\begin{corollary}
Algorithm 2 computes the number of $K_4$ copies in
  in $O(n^{\omega(1,2,1)})$ time while the number of $K_5$ copies
  in $O(n^{\omega(1,3,1)})$ time.
\end{corollary}}

Algorithm 3 can be refined to return the number of $K_r$ copies in
the input graph.  We can also use Fact \ref{fact: wit} to modify
Algorithm 3 to find a copy of $K_r$
in $\tilde{O}(n^{\omega(\lceil \frac {r-q} 2 \rceil,q,\lfloor \frac {r-q} 2 \rfloor )})$ in the graph in case
it contains copies of $K_r.$  

\section{Open problems}
It is an intriguing open problem if the upper bounds in terms of
rectangular matrix multiplication on the $k$-dimensional matrix product
of $k$ square matrices
yielded by Lemma \ref{lem: kdim} are asymptotically tight.
In other words, the question is if
$\omega_k=\min_{k'=1}^{k-1}\omega(k',1,k-k')$ holds
or more specifically 
if $\omega_k=\omega (\lceil k/2
\rceil, 1, \lfloor k/2 \rfloor )$? If this was not the case
for $k$ equal to $3$ or $4$ then we would get a  breakthrough in detection/counting
of $K_4$ or $K_5,$ respectively.

An argument for the inequality $\omega_k< \omega (k_1,1,k_2)$, for positive integers
$k_1,\ k_2$ satisfying $k=k_1+k_2,$ is that in the context of the efficient reduction in the proof of
Lemma \ref{lem: kdim}, the rectangular matrix product seems more general than
the $k$-dimensional one. A reverse efficient reduction seems to be possible only under
very special assumptions. However, proving such an inequality would be extremely
hard as it would imply $\omega (k_1,1,k_2)>k$ and in consequence $\omega>2$
by the straightforward reduction of the rectangular matrix product to the square one.
On the other hand, this does not exclude the possibility of establishing better upper bounds
on $\omega_k$ than those known on  $\omega (k_1,1,k_2)$.

Our alternative methods for detection/counting $K_r$ copies are competitive
and promising for $r\le 5.$ It is also an interesting question if
there is a truly alternative method for detection/counting $K_r$ copies
that could at least match the generalized triangle method for $r>5$?


\begin{thebibliography}{10}

\bibitem{AV21}
Alman, J., Vassilevska Williams, V.:
\newblock
A Refined Laser Method and Faster Matrix Multiplication.
\newblock
Proc. SODA 2021, pp. 522-539.

\bibitem{ADH08}
Alon,~N., Dao,~P., Hajirasouliha,~I., Hormozdiari,~F., Sahinalp,~S.~C.:
Biomolecular network motif counting
and discovery by color coding.
Bioinformatics (ISMB 2008), 24(13), pp.~241--249 (2008)

\bibitem{AN96}
Alon,~N., Naor,~M.: Derandomization, witnesses for Boolean matrix multiplication
and construction of perfect hash functions. Algorithmica 16, 434--449 (1996)

\bibitem{BKS18}
Bl\"aser,M., Komarath,~B., Sreenivasaiah,K.:
Graph Pattern Polynomials. CoRR.abs/1809.08858, 2018.
\junk{
\bibitem{CPS85}
Corneil,~D.G., Perl,~Y., Stewart,~L.K.:
A Linear Recognition Algorithm for Cographs.
\newblock SIAM J. Comput. 14(4), pp. 926--934 (1985)



\bibitem{CG83}
Chung,~F.R.K., Grinstead,~C.M.:
A Survey of Bounds for Classical Ramsey 
Numbers.
Journal of Graph Theory, vol. 7, pp. 25--37, (1983)
}

\bibitem{DDV19}
Dalirrooyfard, M., Duong Vuong, T., Virginia Vassilevska Williams, V.: 
Graph pattern detection: Hardness for all induced patterns and faster non-induced cycles.
Proc. STOC 2019.

\bibitem{EG04}
Eisenbrand,~F., Grandoni,~F.:
On the complexity of fixed parameter clique and
dominating set.
Theoretical Computer Science 326, pp.~57--67 (2004) 

\junk{
\bibitem{ES11}
Eschen,~E.M., Ho\`{a}ng,~C.T., Spinrad,~J., Sritharan,~R.:
On graphs without a C4 or a diamond.
Discrete Applied Mathematics 159(7), pp. 581--587 (2011)
}
\bibitem{FKL15}
Floderus,~P., Kowaluk,~M., Lingas,~A., Lundell,~E.-M.:
Detecting and Counting Small Pattern Graphs.
SIAM J. Discrete
Math. 29(3), pp. 1322--1339 (2015)


\bibitem{FKLL12}{ES11}
Floderus,~P., Kowaluk,~M., Lingas,Ã., Lundell,~E.-M.:
Induced subgraph isomorphism:
Are some patterns substantially easier than others?.
Theoretical Computer Science 605, pp. 119-128 (2015)


\bibitem{LGU}
Le Gall,~F. and Urrutia, F.:
Improved Rectangular Matrix Multiplication using Powers of the Coppersmith-Winograd Tensor.
In: Proc. SODA 2018, pp. 1029--1046 (2018)

\bibitem{HP}
Huang, X., and Pan, V.Y.:
\newblock Fast rectangular matrix multiplications and applications.
\newblock Journal of Complexity, 14, pp. 257--299, 1998.

\junk{
\bibitem{kwitness}
G\k{a}sieniec,~L., Kowaluk,~M., Lingas,~A.: Faster multi-witnesses for Boolean matrix product. 
Information Processing Letters 109, pp. 242--247 (2009)

\bibitem{HKS13}
Ho\`{a}ng,~C.T., Kaminski,~M., Sawada,~J., Sritharan,~R.:
Finding and listing induced paths and cycles.
Discrete Applied Mathematics 161(4-5), pp. 633--641 (2013)}

\bibitem{itai}
Itai,~A., Rodeh,~M..:
Finding a minimum circuit in a graph.
SIAM Journal of Computing, vol. 7, pp. 413--423 (1978)

\bibitem{KKM00}
Kloks,~T., Kratsch,~D., M\"uller,~H.:
Finding and counting small induced subgraphs efficiently.
Information Processing Letters 74(3-4), pp. 115--121 (2000)

\bibitem{KLL13}
Kowaluk,~M., Lingas,~A., Lundell,~E.-M.:
Counting and detecting small subgraphs via equations and matrix multiplication.
SIAM J. on Discrete Mathematics 27(2), pp. 892--909 (2013)
\junk{
\bibitem{KL17}
Kowaluk,~M., Lingas,~A.:
A Fast Deterministic Detection of Small Pattern Graphs
in Graphs without Large Cliques.
Proc. of the 11th International Conference and Workshops  (WALCOM~2017),
LNCS~10167,
pp.~217--227, Springer, Switzerland, 2017.}

\bibitem{LeGall12}
Le Gall,~F.: Faster Algorithms for Rectangular Matrix Multiplication. 
In: Proc. 53rd Symposium on Foundations of Computer Science (FOCS), pp. 514--523 (2012)

\junk{
\bibitem{LG14}
Le Gall,~F.: Powers of Tensors and Fast Matrix Multiplication. 
In: Proc. 39th International Symposium on Symbolic and Algebraic Computation, pp. 296--303 (2014)}

\bibitem{NP85}
Ne\u{s}et\u{r}il,~J., Poljak,~S.:
On the complexity of the subgraph problem.
Commentationes Mathematicae Universitatis Carolinae, 26(2), 
pp. 415--419 (1985)
\junk{
\bibitem{Olariu88}
Olariu,~S.:
Paw-Free Graphs.
Information Processing Letters 28, pp. 53--54, (1988)}

\bibitem{SW05}
Schank,~T., Wagner,~D.:
Finding, Counting and Listing All Triangles in Large Graphs, an
Experimental Study.
In: Proc. WEA, pp. 606--609 (2005)
\junk{
\bibitem{SS98}
Schnorr,~C.P., Subramanian,~C.R.:
Almost Optimal (on the average) Combinatorial
Algorithms for Boolean Matrix Product Witnesses,
Computing the Diameter. 
In: Proc. Randomization and
Approximation Techniques in Computer Science.
Second International Workshop, RANDOM'98,
Lecture Notes in Computer Science 1518, pp. 218--231 (1998)}

\bibitem{SS04}
Sekar,V., Xie,Y., Maltz,D.A., Reiter,M.K,  Zhang,H.:
Toward a framework for internet forensic analysis.
Third Workshop on Hot Topics in Networking (HotNets-HI), 2004.

\bibitem{WK09}
Wolinski,~C., Kuchcinski,~K., Raffin,~E.:
Automatic Design of Application-Specific
Reconfigurable Processor Extensions with UPaK Synthesis Kernel.
ACM Transactions on Design Automation of
Electronic Systems, 15(1), pp. 1--36 (2009)


\bibitem{VW15}
Vassilevska Williams,~V.,  Wang,~J.R., Williams,~R., Yu~H.:
Finding Four-Node Subgraphs in Triangle Time.
In: Proc. of SODA, pp. 1671--1680 (2015)
\junk{
\bibitem{Vassilevska12}
Vassilevska Williams,~V.: Multiplying matrices faster than Coppersmith-Winograd. 
In: Proc. 44th Annual ACM Symposium on Theory of Computing (STOC), pp. 887--898 (2012)

\bibitem{Yu}
Yu, H.:
An Improved Combinatorial Algorithm for Boolean Matrix Multiplication.
In: Proc. 42nd International Colloquium on Automata, Languages and Programming (ICALP),
LNCS 9134, pp. 1094--1105 (2015).}

\end{thebibliography}
\end{document}